\documentclass[11pt]{article}
\usepackage{waingarten}
\usepackage{mathrsfs}

\newcommand{\cT}{\mathscr{T}}
\newcommand{\STAT}{\mathrm{STAT}}

\title{On Mean Estimation for General Norms with Statistical Queries}
\author{Jerry Li \\ Microsoft Research \\ \texttt{jerrl@microsoft.com}
 \and Aleksandar Nikolov \\ University of Toronto \\ \texttt{anikolov@cs.toronto.edu}
 \and Ilya Razenshteyn \\ Microsoft Research \\ \texttt{ilyaraz@microsoft.com} \vspace{0.1in}
  \and Erik Waingarten \\Columbia University\\ \texttt{eaw@cs.columbia.edu}\vspace{0.3cm}}

\begin{document}
\maketitle
\begin{abstract}
We study the problem of mean estimation for high-dimensional distributions given access to a statistical query oracle. For a normed space $X = (\R^d, \|\cdot\|_X)$ and a distribution supported on vectors $x \in \R^d$ with $\|x\|_{X} \leq 1$, the task is to output an estimate $\hat{\mu} \in \R^d$ which is $\eps$-close in the distance induced by $\|\cdot\|_X$ to the true mean of the distribution. We obtain sharp upper and lower bounds for the statistical query complexity of this problem when the the underlying norm is \emph{symmetric}
as well as for Schatten-$p$ norms, answering two questions raised by Feldman, Guzm\'{a}n, and Vempala (SODA 2017).
\end{abstract}
\thispagestyle{empty}

\setcounter{page}{1}


\section{Introduction}

Let $D$ be a distribution over $\mathbb{R}^d$. Informally speaking, in the statistical query model (SQ), one
learns about $D$ as follows. Given a query $h\colon \mathbb{R}^d \to [-1,1]$, the SQ oracle
with tolerance $\tau > 0$ reports $\Ex_{x \sim D} [h(x)]$ perturbed by error of scale roughly $\tau$.
The SQ model was introduced in~\cite{kearns1998efficient} as a way
to capture ``learning algorithms that construct a hypothesis based on statistical
properties of large samples rather than on the idiosyncrasies of a particular sample.''

The original motivation for the SQ framework was to provide an evidence of computational hardness
of various learning problems (beyond sample complexity) by proving
lower bounds on their SQ complexity.
Indeed, many learning algorithms (see~\cite{feldman2016statistical} for an overview)
can be captured by the SQ framework, and, furthermore, the only known technique that gives
a polynomial-time algorithm for a learning problem with exponential SQ complexity~\cite{kearns1998efficient} is \emph{Gaussian 
elimination over finite fields}, whose utility for learning is currently extremely limited.
This reasoning suggests the following heuristic:
\begin{center}
    \emph{If solving a learning problem to accuracy $\varepsilon > 0$ requires $d^{\omega(1)}$ SQ queries with tolerance $\varepsilon^{O(1)} / d^{O(1)}$, then it is unlikely to be doable in time $d^{O(1)}$
    using \emph{any} algorithm.}
\end{center}
This heuristic together with the respective SQ lower bounds provided strong evidence of hardness of
many problems such as: learning parity with noise~\cite{kearns1998efficient}, learning intersection of half-spaces~\cite{klivans2007unconditional}, the planted
clique problem~\cite{feldman2013statistical}, robust estimation of high-dimensional Gaussians and non-Gaussian component analysis~\cite{diakonikolas2017statistical}, learning a small neural network~\cite{song2017complexity},
adversarial learning~\cite{bubeck2018adversarial}, robust linear
regression~\cite{diakonikolas2019efficient}, among others.

However, over time, the SQ model has generated significant \emph{intrinsic} interest~\cite{feldman2016general}, in part
due to the connections to distributed learning~\cite{steinhardt2016memory} and local differential privacy~\cite{kasiviswanathan2011can}. In particular, the new goal is to understand the trade-off between
the number and the tolerance of SQ queries, and the accuracy of the resulting solution for various learning problems, which is \emph{more nuanced} than what is necessary for the above ``crude'' heuristic.
In a paper by Feldman, Guzman, and Vempala~\cite{FGV17}, this was done for
perhaps the most basic learning problem, \emph{mean estimation}, which is formulated as follows.

\begin{problem}[Mean estimation using statistical queries]
\label{main_prob_intro}
Let $D$ be a distribution over the \emph{unit ball} $B_X$ of a normed space $X = (\mathbb{R}^d, \|\cdot\|_X)$,
and suppose we are allowed $d^{O(1)}$ statistical queries with tolerance $\varepsilon > 0$. What is the smallest
$\varepsilon' > 0$, for which we can always recover a point $\widehat{x}$ such that
$\left\|\widehat{x} - \Ex_{\bx \sim D}[\bx]\right\|_{X} \leq \varepsilon'$
holds w.h.p.
\end{problem}
Clearly, $\varepsilon' \geq \varepsilon$, and, as~\cite{FGV17} showed, $\varepsilon' \leq O(\varepsilon \sqrt{d})$ for every norm. We say that a norm $\|\cdot\|$ over $\mathbb{R}^d$
is \emph{tractable} if one can achieve $\varepsilon' \leq \varepsilon \cdot \mathrm{poly}(\log d, \log(1/\varepsilon) )$ (with $\mathrm{poly}(d)$ queries of tolerance $\varepsilon$).
The main result of~\cite{FGV17} can be stated as follows.
\begin{theorem}[\cite{FGV17}]
    \label{vitaly_thm_intro}
    The $\ell_p$ norm over $\mathbb{R}^d$ is tractable if and only if $p \geq 2$.
\end{theorem}
The fact that the $\ell_\infty$ norm is tractable is trivial, since we can estimate
each coordinate of the mean separately. However, the corresponding algorithm for $\ell_p$
norms for $2 \leq p < \infty$ is more delicate and is based on random rotations,
while the na\"{\i}ve coordinate-by-coordinate estimator merely gives
$\varepsilon' = \varepsilon  d^{\Theta_p(1)}$.
\cite{FGV17} raise several intriguing open problems, among them
the following two:
\begin{enumerate}
    \item Characterize tractable norms beyond $\ell_p$;
    \item Solve Problem~\ref{main_prob_intro} for the spectral norm and other Schatten-$p$ norms of matrices;
\end{enumerate}
In this paper, we make progress towards solving the first problem and completely resolve the second one.

\subsection{Our results}

\paragraph{Symmetric norms.} Our first result gives a complete characterization of \emph{symmetric}
tractable norms. A norm is symmetric if it is invariant under all permutations of coordinates and sign flips (for many examples beyond $\ell_p$ norms, see~\cite{andoni2017approximate}). Recently there has been substantial progress in understanding various algorithmic tasks for general symmetric norms~\cite{blasiok2017streaming,andoni2017approximate,song2018towards,andoni2018subspace}.
In this paper, we significantly extend Theorem~\ref{vitaly_thm_intro} to all the symmetric norms.
To formulate our result, we need to define the \emph{type-$2$ constant} of a normed space, which is one of the standard
bi-Lipschitz invariants (\cite{wojtaszczyk1996banach}).

\begin{definition}
    For a normed space $X = (\mathbb{R}^d, \|\cdot\|_X)$, the type-$2$
    constant of $X$, denoted by $T_2(X)$, is defined as the smallest $T > 0$ such that the following holds.
    For every sequence of vectors $x_1, x_2, \ldots, x_n \in X$ and for uniformly random
    $\beps \sim \{-1, 1\}^n$, one has:
    \begin{align}
    \left(\Ex_{\beps \sim \{-1,1\}^n} \left\|\sum_{i=1}^n \beps_i x_i\right\|_X^2\right)^{1/2}
    \leq T \cdot \left(\sum_{i=1}^n \|x_i\|_X^2\right)^{1/2}. \label{eq:type}
    \end{align}
\end{definition}
We are now ready to state our result.
\begin{theorem}
    \label{our_symm_intro}
    A symmetric normed space $X = (\mathbb{R}^d, \|\cdot\|_X)$ is tractable
    iff $T_2(X) \leq \mathrm{poly}(\log d)$.
\end{theorem}
Theorem~\ref{our_symm_intro} easily implies Theorem~\ref{vitaly_thm_intro}, since for $1 \leq p < 2$,
$T_2(\ell_p) = d^{\Omega_p(1)}$, while for $2 \leq p < \infty$ one has $T_2(\ell_p) \leq \sqrt{p - 1}$
and $T_2(\ell_\infty) \leq O(\sqrt{\log d})$~(\cite{BCL94}). For a quantitative version of Theorem~\ref{our_symm_intro},
see Theorem~\ref{thm:sym-norm-ub} and Theorem~\ref{thm:symm-lb}.

\paragraph{Schatten-$p$ norms.} Recall that for a matrix $M$, the Schatten-$p$ norm of $M$ is the $\ell_p$ norm of the singular values of $M$.
In particular, the Schatten-$\infty$ norm of $M$ is simply the spectral norm of $M$, and the Schatten-$2$ norm corresponds to the Frobenius norm.
Such norms are very well-studied and arise naturally in many applications in learning and probability theory.
Our second main result settles the tractability of Schatten-$p$ norms, resolving a question of~\cite{FGV17}.
\begin{theorem}
    \label{our_schatten_intro}
    The Schatten-$p$ norm is tractable iff $p = 2$.
\end{theorem}
For a quantitative version of Theorem~\ref{our_schatten_intro}, see Theorem~\ref{schatten_main}. 
Theorem~\ref{our_schatten_intro} shows that one cannot remove ``symmetric'' from Theorem~\ref{our_symm_intro}, since type-$2$ constants of Schatten-$p$ spaces are essentially the same as for the corresponding $\ell_p$ spaces~(\cite{BCL94}).
Specifically, for $p > 2$, Schatten-$p$ spaces have small type-$2$ constant, but are intractable. 
In particular, we show that the best
mean estimation algorithm for Schatten-$p$ can be obtained by embedding the space into $\ell_2$ (via the identity
map) and then using the $\ell_2$ estimation algorithm from~\cite{FGV17}.

\subsection{Techniques}

The main technical tool underlying the algorithm for mean estimation in symmetric norms is the following geometric statement. For any symmetric norm $(\R^d, \|\cdot\|_X)$, consider the set $R_j \subset B_{X}$ consisting of the \emph{level-$j$ ring}, i.e., all points $x \in B_{X}$ whose non-zero coordinates have absolute value between $2^{-(j+1)}$ and $2^{-j}$, and consider the smallest radius $r > 0$ where $R_j \subset r B_{\ell_2}$. Then,
\begin{align} 
R_j \subset r B_{\ell_2} \cap 2^{-j} B_{\ell_{\infty}} \subset (3 T_2(X) \log d) B_X. \label{eq:inclusion}
\end{align}
Given the above geometric statement, which generalizes the similar statement for $\ell_p$ norms from \cite{FGV17}, we generalize the algorithm from \cite{FGV17} to the symmetric norms setting. Specifically, we divide the distribution into $\log(d/\eps)$ distributions, each lying on a \emph{level-$j$ ring} of $B_X$, so that the sum of the estimates of the $\log(d/\eps)$ distributions is a good estimate for the original distribution. By the first inclusion in (\ref{eq:inclusion}), we may use the mean estimation algorithms for $\ell_{\infty}$ and $\ell_2$ on each ring after an appropriate scaling with error $\eps$. Running these two algorithms, we can get an approximation to mean of the distribution on the ring up to error $O(\eps 2^{-j})$ in $\ell_{\infty}$ and $\eps r$ in $\ell_2$. Via the second inclusion in (\ref{eq:inclusion}), this will be a good estimate in $X$ provided $T_2(X)$ is small.

The lower bound for norms with large type-$2$ constants is a generalization of the result in~\cite{FGV17}; in particular, the hard distributions for $\ell_p$ from \cite{FGV17} are supported on basis vectors, which are exactly those achieving $T_2(\ell_p)$ in (\ref{eq:type}). For general norms $X$, we consider the analogous distributions supported on an arbitrary set of vectors achieving $T_2(X)$ in (\ref{eq:type}); however, the fact that we have much less control on the vectors necessitates additional care.

The Schatten-$p$ norms, for $p > 2$, do satisfy $T_2(S_p) \leq \sqrt{\log d}$, so new ideas are required in proving the lower bound. We show the lower bound for carefully crafted hard distributions, using hypercontractivity to show concentration of the result of an arbitrary statistical query.

\section{Preliminaries}

Here we introduce some basic notions about normed spaces and
statistical algorithms. 
We will use boldfaced letters for random variables, and the notation $\beps \sim \{-1,1\}^n$ will mean that
$\beps$ is a random vector chosen uniformly from $\{-1,1\}^n$. 

\begin{definition}
For any vector $x \in \R^d$, we let $|x|$ be the vector $x$ with each
coordinate replaced by its absolute value, and let $x^* = P|x|$ be the
vector obtained by applying the permutation matrix $P$ to $|x|$ which
sorts coordinates of $|x|$ by order of non-increasing value. A normed
space $X = (\R^d, \|\cdot\|_X)$ is symmetric if $\|x\|_X = \|x^*\|_X$
holds for every $x \in \R^d$. 
\end{definition}

We recall that $\ell_p^d$ is the normed space over $\R^d$ with the
norm of a vector $x$ given by $\|x\|_p = (|x_1|^p + \ldots +
|x_d|^p)^{1/p}$. The Schatten-$p$ space $S^d_p = (\R^d, \|\cdot\|_{S_p})$ is
defined over $d\times d$ matrices with real entries, and the norm of a
matrix is defined as the $\ell^d_p$ norm of its singular values. We
omit the superscript $d$ and just write $\ell_p$ and $S_p$
when this does not cause confusion.

For a normed space $X = (\R^d, \|\cdot\|_X)$, let $B_X = \{ x \in \R^d : \|x\|_X \leq 1\}$ be the unit ball of the norm $X$. Furthermore, for $p \in [1, \infty)$, we let $L_p(X) = (\R^{dn}, \|\cdot\|_{L_p(X)})$ be the normed space over sequences of vectors $x = (x_1, \dots, x_n) \in \R^{d n}$ where $\|x\|_{L_p(X)} = ( \sum_{i=1}^n \|x_i\|_X^p )^{1/p}.$

Next we define the type of a normed space.
\begin{definition}
Let $X = (\R^d, \|\cdot\|_X)$ be a normed space, $n \in \N$, and $p \in [1, 2]$. Let $T_p(X, n)$ be the infimum over $T > 0$ such that:
\[ \left( \Ex_{\beps \sim \{-1,1\}^n} \left[ \left\| \sum_{i=1}^n \beps_i x_i \right\|_X^2 \right]\right)^{1/2} \leq T \left( \sum_{i=1}^n \|x_i\|_X^p \right)^{1/p}, \]
for all $x_1, \dots, x_n \in \R^d$. We let
$T_p(X) = \sup_{n \in \N} T_p(X, n)$, and say $X$ has \emph{type} $p$ with constant $T_p(X)$.
\end{definition}
Note that, by the parallelogram identity, the Euclidean space $(\R^d,
\|\cdot\|_2)$ has type $2$ with constant $1$, and in fact the
inequality becomes an equality. Together with John's theorem, this
implies that any $d$-dimensional normed space has type $2$ with
constant at most $\sqrt{d}$. However, we are typically interested in
spaces that have type $p$ with constant independent of dimension. It
follows from the results in~\cite{BCL94} that for $p \ge 2$,
$\ell^d_p$ has type $2$ with constant $\sqrt{p-1}$, and
for $1 \le p < 2$, $\ell_p^d$ has type $p$ with constant
$1$; at the same time, considering the standard basis of $\R^d$ shows
that for $1 \le p < q \le 2$, the type $q$ constant of $\ell_p^d$ goes
to infinity with the dimension $d$. Moreover, these results also hold
for Schatten-$p$ spaces.

For a normed space $X = (\R^d, \|\cdot\|_X)$, let $B_X = \{ x \in \R^d : \|x\|_X \leq 1\}$ be the unit ball of the norm $X$. Furthermore, for $p \in [1, \infty)$, we let $L_p(X)$ be the normed space over sequences of vectors $x = (x_1, \dots, x_n) \in \R^{d \cdot n}$ where $\|x\|_{L_p(X)} = ( \sum_{i=1}^n \|x_i\|_X^p)^{1/p}$.

\newcommand{\SDN}{\mathrm{SDN}}
\newcommand{\VSTAT}{\mathrm{VSTAT}}

Finally, we define formally statistical algorithms and the $\STAT$ and
$\VSTAT$ oracles. We follow the definitions from~\cite{FGRVX13}.
\begin{definition}
  Let $D$ be a distribution supported on $\Omega$. For a tolerance
  parameter $\tau>0$, the oracle $\STAT(\tau)$ takes a query function
  $h\colon\Omega \to [-1,1]$, and returns some value $v \in \R$ satisfying
  $|v - \Ex_{\bx\sim D}[h(\bx)]|\leq \tau$.
  For a sample size parameter $t>0$, the $\VSTAT(t)$ oracle takes a
  query function $h\colon\Omega\to[0,1]$ and returns some value $v \in \R$
  such that $|v - p| \leq\tau$, for $p=\Ex_{\bx\sim D}[h(\bx)]$, and
  $\tau = \max\{1/t, \sqrt{p(1-p)/t}\}$.
  
  We call an algorithm that accesses the distribution $D$ only via one
  of the above oracles a \emph{statistical algorithm}.
\end{definition}
Clearly, $\VSTAT(t)$ is at least as strong as $\STAT(1/\sqrt{t})$ and no
stronger than $\STAT(1/t)$. 
The lower bounds presented will follow the framework of \cite{FPV18}.
\begin{definition}\label{def:stats}
The discrimination norm $\kappa_2(D, \calD)$ for a distribution $D$ supported on $\Omega$ and a set $\calD$ of distributions supported on $\Omega$ is given by:
\[ \kappa_2(D, \calD) = \max_{\substack{h \colon \Omega \to \R \\ \|h\|_{D} = 1}} \left\{ \Ex_{\bD \sim \calD}\left[ \left| \Ex_{\bx \sim D}[h(\bx)] - \Ex_{\bx \sim \bD}[h(\bx)] \right| \right] \right\}, \]
where $\bD \sim \calD$ is sampled uniformly at random, and $\|h\|_{D}^2 = \Ex_{\by \sim D}[h(\by)^2]$. The decision problem $\calB(D, \calD)$ is the problem of distinguishing whether an unknown distribution $\bH = D$ or is sampled uniformly from $\calD$. The statistical dimension with discrimination norm $\kappa$, $\SDN(\calB(D, \calD), \kappa)$, is the largest integer $t$ such that for a finite subset $\calD_D \subset \calD$, any subset $\calD' \subset \calD_D$ of size at least $|\calD_D|/t$ satisfies $\kappa_2(\calD', D) \leq \kappa$.
\end{definition}

\begin{theorem}[Theorem 7.1 in \cite{FPV18}]\label{thm:fpv}
For $\kappa > 0$, let $t = \SDN(\calB(D, \calD), \kappa)$ for a distribution $D$ and set of distributions $\calD$ supported on a domain $\Omega$. Any randomized statistical algorithm that solves $\calB(D, \calD)$ with probability at least $2/3$ requires $t/3$ calls to $\VSTAT(1/(3\kappa^2))$.
\end{theorem}


\section{Symmetric norms}

\subsection{Mean estimation using SQ for type-2 symmetric norms}

\begin{definition}
Let $X = (\R^d, \|\cdot\|_X)$ be any symmetric norm with $\|e_1\|_X = 1$. Let $\ell_X \colon (0, 1] \to \{0, 1, \dots, d\}$ be the maximum number of coordinates set to $t$ in a vector within the unit ball of $X$, i.e.,
\[ \ell_X(t) = \max\left\{ k : \|(\underbrace{t, \dots, t}_{k}, 0, \dots, 0) \|_{X} \leq 1\right\}, \]
and $m_X \colon (0, 1] \to \R^{\geq 0}$ be the maximum $\ell_2$ norm of a vector within the unit ball of $X$ with coordinates set to $t$, i.e.,
\[ m_X(t) = \max\left\{ \|x\|_2 : x = (\underbrace{t, \dots, t}_{\leq \ell_X(t)}, 0, \dots, 0) \right\}. \]
\end{definition}

The following is the main lemma needed for the statistical query algorithm for type-2 symmetric norms. The lemma is a generalization of Lemma 3.12 from \cite{FGV17} from $\ell_p$ norms (with $p > 2$) to arbitrary type-$2$ symmetric norms. The lemma bounds the norm in $X$ of an arbitrary vector $x$, given corresponding bounds on the $\|x\|_{\infty}$ and $\|x\|_2$. 

\begin{lemma}\label{lem:inter}
Let $X = (\R^d, \|\cdot\|_X)$ be a symmetric norm with type-$2$ constant $T_2(X) \in [1, \infty)$. Fix any $t \in (0, 1]$, and let $x \in \R^d$ satisfy $\|x\|_{\infty} \leq t$ and $\|x\|_2 \leq m_X(t)$. Then, $\|x\|_X \leq T_2(X) \cdot 3\log d$.
\end{lemma}

\begin{proof}
Given the vector $x \in \R^d$, consider the sets $B_j(x) \subset [d]$ for $j \in \{ 0, \dots, 2\log(d)\}$ given by
\[ B_j(x) = \left\{ i \in [d] : t \cdot 2^{-j-1} < |x_i| \leq t \cdot 2^{-j} \right\}, \]
and let $x^{(j)} \in \R^d$ be the vector given by letting the first
$|B_j(x)|$ coordinates be $t \cdot 2^{-j}$, and the remaining
coordinates be 0. Because $X$ is symmetric with respect to changing
the sign of any coordinate of $x$, the triangle inequality easily
implies that $\|x\|_X$ is monotone with respect to $|x_i|$ for any $i
\in [d]$. Then, by the triangle inequality and the fact that $X$ is symmetric with $\|e_1\|_X = 1$, $\|x\|_X \leq \sum_{j = 0}^{2\log(d)} \|x^{(j)}\|_X + t/d$; thus, it remains to bound $\|x^{(j)}\|_X$ for every $j \in \{0, \dots, 2\log(d)\}$. 

We then have
$\sqrt{|B_j(x)|} \cdot t \cdot 2^{-j} = \|x^{(j)}\|_2 \leq m_X(t) \leq t\sqrt{\ell_X(t)}$,
where, in the first inequality, we used the fact that $\|x^{(j)}\|_2
\leq \|x\|_2 \leq m_X(t)$, and, in the second inequality, we used the
definition of $\ell_X(t)$. As a result, we have $|B_j(x)| \leq
\ell_X(t) \cdot 2^{2j}$, so consider partitioning the non-zero
coordinates of $x^{(j)}$ into at most $s=2^{2j}$ groups, each of size
at most $\ell_X(t)$, and let $v_1, \dots, v_{s} \in \R^d$ be the
vectors so $x^{(j)} = \sum_{i=1}^s v_s$. We have
\begin{align*}
\|x^{(j)}\|_X^2 = \Ex_{\beps \sim\{-1,1\}^s}\left[ \left\| \sum_{i=1}^s \beps_i v_i \right\|_X^2 \right] \stackrel{(a)}{\leq} T_2(X)^2 \sum_{i=1}^s \|v_i\|_X^2 \stackrel{(b)}{\leq} T_2(X)^2,
\end{align*}
where the equality uses the symmetry of $X$ with respect to changing
signs of coordinates, the inequality (a) uses the definition of type
constants, and the inequality (b) follows from the definition of $\ell_X(t)$.
 We obtain the desired lemma by summing over all $\|x^{(j)}\|_X$, for $j \in \{0, \dots, 2\log(d)\}$. 
\end{proof}
\noindent
With this structural result, we now show:
\begin{theorem}\label{thm:sym-norm-ub}
Let $X = (\R^d, \|\cdot\|)$ be a symmetric norm with type-2 constant $T_2(X) \in [1, \infty)$ normalized so $\|e_1\|_X = 1$. There exists an algorithm for mean estimation over $X$ making $3d\log d$ queries to $\STAT(\alpha)$, where the accuracy $\alpha$ satisfies 
\[ \alpha = \Omega\left(\frac{\eps}{T_2(X) \cdot \log d \cdot \log(d/\eps))}\right). \]
\end{theorem}

\begin{proof}
For $j \in \{0, \dots, 2\log(d/\eps)\}$, and $w \in \R^d$, let $R_j(w)$ be the level $j$ vector of $w$, i.e., $R_j(w) = \sum_{i=1}^n e_i w_i \ind\{ w_i \in (2^{-j-1}, 2^j]\}$. For any fixed distribution $\calD$ supported on the unit ball of $X$, we may consider the distribution $\calD_j$ given by $R_j(\bx)$ where $\bx \sim \calD$. Denote $\mu = \Ex_{\bx \sim \calD}[\bx]$ and $\mu_j = \Ex_{\bx \sim \calD_j}[\bx]$, so that distributions $\calD_j$ satisfy $\| \mu - \sum_j \mu_j \|_X \leq \eps^2/d$. As a result, the sum of $\eps/(3\log(d/\eps))$-approximations of $\mu_j$ would result in an $\eps$-approximation of $\mu$.

The algorithm proceeds by estimating the mean of each distribution $\calD_j$ and then taking the sum of all estimates:
\begin{enumerate}
\item For each $j \in \{0,\dots, 2\log(d/\eps)\}$, we consider $\calH_{\infty}^{(j)}$ as the distribution given by $\bx / 2^{-j}$ where $\bx \sim \calD_j$, and $\calH_{2}^{(j)}$ as the distribution given by $\bx / (2m_X(2^{-j}))$. Note that $\calH_{\infty}^{(j)}$ is supported on $B_{\ell_\infty}$, and $\calH_2^{(j)}$ is supported on $B_{\ell_2}$. 
\begin{itemize}
\item Perform the mean estimation algorithms for $\calH_{\infty}^{(j)}$ and $\calH_{2}^{(j)}$ as given in \cite{FGV17} with error parameter $\eps \gamma$ where $\gamma = 1 / (36 \cdot T_2(X) \cdot \log d\log(d/\eps))$ to obtain vectors $v_{\infty}^{(j)}, v_{2}^{(j)} \in \R^d$, and let $w_{\infty}^{(j)} = 2^{-j} v_{\infty}^{(j)}$ and $w_2^{(j)} = 2 m_X(2^{-j}) v_{2}^{(j)}$ where 
\begin{align} \label{eq:estimates}
\left\| \mu_j - w_{\infty}^{(j)} \right\|_\infty \leq \eps \gamma \cdot 2^{-j} \qquad \text{and}\qquad \left\| \mu_j - w_{2}^{(j)} \right\|_2 \leq 2\eps\gamma \cdot m_X(2^{-j}). 
\end{align} 
\item Find one vector $w^{(j)} \in \R^d$ where $\|w^{(j)} - w_{\infty}^{(j)}\|_{\infty} \leq \eps\gamma 2^{-j}$ and $\|w^{(j)} - w_{2}^{(j)}\|_2 \leq 2\eps\gamma m_X(2^{-t})$, and return $w^{(j)}$ as an estimate for $\mu_j$.
\end{itemize}
\item Given estimates $w^{(j)} \in \R^d$ for all $j \in \{0,\dots, 2\log(d/\eps)\}$, output $\sum_j w^{(j)}$.
\end{enumerate}

We note that the inequalities in (\ref{eq:estimates}) follow from the fact that $v_{\infty}^{(j)}$ and $v_2^{(j)}$ are $\eps \gamma$-approximations for $\Ex_{\bx \sim \calH_{\infty}^{(j)}}[\bx]$ (in $\ell_{\infty}$) and $\Ex_{\bx \sim \calH_{2}^{(j)}}[\bx]$ (in $\ell_2$), respectively, and that 
\[ 2^{-j} \Ex_{\bx \sim \calH_{\infty}^{(j)}}[\bx] = 2m_X(2^{-j}) \Ex_{\bx \sim \calH_2^{(2)}}[\bx] = \mu_j.\] 
In order to see that $w^{(j)}$ is a good estimate for $\mu_j$, let $y_j = \mu_j - w^{(j)}$ be the error vector in the approximation. From the triangle inequality, and the definition of $w^{(j)}$, we have $\|y\|_{\infty} \leq 2\eps\gamma \cdot 2^{-j}$ and $\|y\|_2 \leq 4\eps\gamma \cdot m_X(2^{-j})$, so that Lemma~\ref{lem:inter} implies $\|y\|_X \leq 12\eps \gamma \cdot T_2(X) \log d \leq \eps / (3\log(d/\eps))$.
\end{proof}

\subsection{Lower bounds for normed spaces with large type-2 constants}\label{sec:type-2-lbs}

We now give a lower bound for normed spaces which have large type-2 constant.

\begin{theorem}\label{thm:symm-lb}
Let $X = (\R^d, \|\cdot\|_X)$ be a normed space with type-2 constant $T_2(X) \in [1, \infty)$. There exists an $\eps > 0$ such that any statistical algorithm for mean estimation in $X$ with error $\eps$ making queries to $\VSTAT(1/(3\kappa^2))$ must make 
\[ \exp\left(\Omega\left(\frac{T_2(X)^2 \cdot \kappa^2}{\eps^2 \cdot \log d}\right)\right) \]
 such queries.
\end{theorem} 
\noindent
The immediate corollary of Theorem~\ref{thm:symm-lb} shows the upper bound from Theorem~\ref{thm:sym-norm-ub} is tight up to poly-logarithmic factors.
\begin{corollary}
Let $X = (\R^d, \|\cdot\|_X)$ be a normed space with type-$2$ constant $T_2(X) \in [1, \infty)$.  Any algorithm for mean estimation in $X$ making $d^{O(1)}$-queries to $\VSTAT(\alpha)$ must have 
\[ \alpha = O\left( \frac{\eps \cdot \log d}{T_2(X)}\right). \] 
\end{corollary}
\noindent
We set up some notation and basic observations leading to a proof of Theorem~\ref{thm:symm-lb}. 

\begin{lemma}\label{lem:type-2}
Let $X = (\R^d, \|\cdot\|_X)$ be a normed space with type-$2$ constant $T_2(X) \in [1, \infty)$. Then, for any $t < T_2(X)$, there exists some $n \in \N$, as well as a sequence of vectors $x = (x_1, \dots, x_n) \in (\R^d)^n$, where $1 \leq \|x_i\|_{X} \leq 2$ for every $i \in [n]$, and 
\begin{align} 
\left( \Ex_{\beps \sim \{-1,1\}^n} \left[ \left\| \sum_{i=1}^n \beps_i x_i\right\|_X^2 \right]\right)^{1/2} &\geq t_2(x) \left( \sum_{i=1}^n \|x_i\|_X^2 \right)^{1/2}\label{eq:type-2-ineq}
\end{align}
with $t_2(x) > t / C$ for an absolute constant $C$.
\end{lemma}
\begin{proof}
  Since $t < T_2(X)$, there exists a sequence $x' = (x'_1,
  \ldots, x'_m)$ such that
  $\Ex_{\beps \sim\{-1,1\}^m} [ \| \sum_{i=1}^m \beps_i x'_i\|_X^2] \geq t \sum_{i=1}^m \|x'_i\|_X^2$.
  A well-known comparison inequality between Rademacher and Gaussian
  averages (see e.g.~Lemma 4.5.~in~\cite{LedouxTalagrand}) gives that
  for a sequence of independent standard Gaussian random variables
  $\bg_1, \ldots, \bg_m$, 
  $\Ex_{\beps} [ \| \sum_{i=1}^m \beps_i x'_i\|_X^2 ]
  \le 
  (\pi/2)\Ex_{\bg} [ \| \sum_{i=1}^m \bg_i x'_i\|_X^2]$.
  Let us assume, without loss of generality, that $\|x'_i\|_X \geq 1$ for every $i \in [n]$. For any $x'_i$, define the sequence
  $x'_{i,1}, \ldots, x'_{i,m_i}$ to consist of $\lfloor
  \|x'_i\|^2_X\rfloor-1$ copies of $x'_i / \|x'_i\|_X$ and a
  single copy of $(1+ \|x'_i\|^2_X - \lfloor
    \|x'_i\|^2_X\rfloor)^{1/2}\cdot x'_i/\|x'_i\|_X$, and note $1 \leq \| x'_{i,j} \|_X \leq 2$ for every $i \in [n]$ and $j \in m_i$. Observe also that, if $\bg_{i,1}, \ldots, \bg_{i,m_i}$ are
  independent standard Gaussian random variables, then
  $\sum_{j=1}^{m_i}{\bg_{i,j} x'_{i,j}}$ is distributed identically to
  $\bg_i x'_i$, and, moreover, $\sum_{j=1}^{m_i}{\|x'_{i,j}\|_X^2} =
  \|x'_i\|_X^2$. Therefore, we have
  $
  \Ex_{\bg} [ \| \sum_{i=1}^m\sum_{j = 1}^{m_i} \bg_{i,j} x'_{i,j}\|_X^2]
  \ge 
  (2t/\pi)
   \sum_{i=1}^m \sum_{j = 1}^{m_i}\|x'_{i,j}\|_X^2.
  $
  By the Gaussian version of the Khinntchine-Kahane inequalities
  (Corollary 3.2.~in~\cite{LedouxTalagrand}) and the Zygmund-Paley
  inequality, we have that for some absolute constant $C'$, with
  probability at least $\frac12$, 
  $
  \| \sum_{i=1}^m\sum_{j = 1}^{m_i} \bg_{i,j} x'_{i,j}\|_X^2
  \ge 
  (t/C') \sum_{i=1}^m \sum_{j = 1}^{m_i}\|x'_{i,j}\|_X^2$.
  
  We define the sequence $x= (x_1, \ldots, x_n)$ to contain $N$ copies
  of each vector $x_{i,j}$, for some large enough integer $N$. By the
  central limit theorem, as $N \to \infty$, $\frac{1}{\sqrt{N}}\sum_{i
    = 1}^n{\beps_i x_i}$ converges in disribution to
  $\sum_{i=1}^m\sum_{j = 1}^{m_i} \bg_{i,j} x'_{i,j}$. Then, for a
  large enough $N$, with probability at least $1/4$, we have that
  $
  \| \sum_{i = 1}^n{\beps_i x_i}\|_X^2
  \ge 
  (t/C')\cdot \sum_{i=1}^n \|x_i\|^2_X.
  $
  The lemma follows with $C = 4C'$, since the left hand side above is always
  non-negative. 
\end{proof}

\paragraph{Description of the lower bound instance}
In this section we describe the instance which achieves the lower bound in Theorem~\ref{thm:symm-lb}.

Fix a sequence $x = (x_1, \dots, x_n) \in (\R^d)^{n}$ satisfying
(\ref{eq:type-2-ineq}) guaranteed to exists by Lemma~\ref{lem:type-2},
and let the sequence $\hat{x} = (\hat{x}_1, \dots, \hat{x}_n) \in
(B_X)^n$ be defined by $\hat{x}_i = x_i / \|x_i\|_X$. In the language of \cite{FGV17}, let $D$ be the \emph{reference} distribution supported on $B_X$ given by sampling $\by \sim D$ where for all $i \in [n]$,
\begin{align}
 \Prx_{\by \sim D}\left[ \by = \hat{x}_i \right] = \Prx_{\by \sim D}\left[\by = -\hat{x}_i \right] = \frac{1}{2} \cdot \frac{\|x_i\|_X}{\|x\|_{L_1(X)}}, \label{eq:reference}
 \end{align}
so that $\mu_0 = \Ex_{\by \sim \calD}[\by] = 0 \in \R^d$.
We will let $\eps_0$ be so that $\eps_0 \leq t_2(x) \cdot \|x\|_{L_2(X)} / \|x\|_{L_1(X)}$.
For $z \in \{-1,1\}^n$, let $D_z$ be the distribution supported on $B_X$ given by sampling $\by \sim D_z$ where for all $i \in [n]$,
\begin{align} 
\Prx_{\by \sim D_z}\left[ \by = \hat{x}_i \right] &= \frac{\|x_i\|_X}{\|x\|_{L_1(X)}} \cdot \left(\frac{1}{2} + \frac{z_i \eps_0}{2 \cdot t_2(x)} \cdot \frac{\|x\|_{L_1(X)} }{\|x\|_{L_2(X)}} \right) \nonumber \\
 \Prx_{\by \sim D_z}\left[ \by = -\hat{x}_i \right] &= \frac{\|x_i\|_X}{\|x\|_{L_1(X)}} \cdot \left(\frac{1}{2} - \frac{z_i \eps_0}{2 \cdot t_2(x)} \cdot \frac{\|x\|_{L_1(X)} }{\|x\|_{L_2(X)}} \right). \label{eq:hard-ref}
\end{align}
Then, 
\begin{align}\label{eq:expectation}
\mu_z \eqdef \Ex_{\by \sim D_z}\left[ \by \right] &= \frac{\eps_0}{t_2(x) \|x\|_{L_2(X)}}\sum_{i=1}^n z_i x_i. 
\end{align}
Consider the distribution $\calD$ on distributions which is uniform
over all $D_z$ where $z \in \{-1,1\}^n$. Then, 
we have\footnote{Here and in the rest of the paper we use $A \gsim B$
  to mean that there exists an absolute constant $C > 0$, independent of
  all other parameters, such that $A \ge B/C$, and, analogously, $A
  \lsim B$ to mean $A \le CB$}:
\begin{align} 
\Ex_{\bz \sim \{-1,1\}^n}\left[ \left\| \mu_{\bz} \right\|_X \right] &= \frac{\eps_0}{t_2(x) \|x\|_{L_2(X)}} \Ex_{\beps \sim \{-1,1\}^n} \left[ \left\| \sum_{i=1}^n \beps_i x_i\right\|_X \right] \label{eq:first-1}\\
&\gsim \frac{\eps_0}{t_2(x) \|x\|_{L_2(X)}} \left(\Ex_{\beps \sim
    \{-1,1\}^n}\left[ \left\| \sum_{i=1}^n \beps_i x_i \right\|_X^2
  \right]\right)^{1/2} =\eps_0, \label{eq:first-2} \\
\Ex_{\bz \sim \{-1,1\}^n}\left[ \|\mu_{\bz}\|_X^2 \right] &= \eps_0^2. \label{eq:first-3}
\end{align}
where (\ref{eq:first-2}) and (\ref{eq:first-3})  follow from the Khintchine-Kahane inequalities and  the definition of $t_2(x)$. By the Payley-Zygmund inequality, $\Prx_{\bz \sim\{-1,1\}^n}\left[ \|\mu_{\bz}\|_X \geq \eps \right] = \Omega(1)$,
for some $\eps = \Omega(\eps_0)$. We thus conclude the following lemma, which follows from the preceding discussion.
\begin{lemma}\label{lem:dist}
Suppose there exists a statistical algorithm for mean estimation over $X$ with error $\eps$ making $q(\eps)$ queries to $\VSTAT(\alpha(\eps))$, then for distribution $D$ as in (\ref{eq:reference}) and set $\calD$ as in (\ref{eq:hard-ref}), $\calB(D, \calD)$ has a statistical algorithm making $q(\eps)$ queries of accuracy $\VSTAT(\alpha(\eps))$ which succeeds with constant probability.
\end{lemma}

We now turn to computing the statistical dimension of $\calB(D, \calD)$, as described in Definition~\ref{def:stats}.

\begin{lemma}\label{lem:statis-dim}
Let $D$ and $\calD$ be the distribution and the set over $B_X$ defined in (\ref{eq:reference}) and (\ref{eq:hard-ref}). For $\kappa > 0$, 
$\SDN(\calB(D, \calD), \kappa) \geq \exp( \Omega(\frac{\kappa^2  t_2(x)^2}{\eps^2}))$.
\end{lemma}

\begin{proof}
Let $h \colon B_X \to \R$ be any function with $\|h\|_D = 1$. Note that
\begin{align*}
\Ex_{\by \sim D_z}\left[ h(\by) \right] - \Ex_{\by \sim D}\left[ h(\by)\right] &= \frac{\eps_0}{2t_2(x) \cdot \|x\|_{L_2(X)}} \sum_{i=1}^n z_i \|x_i\|_X \left( h(\hat{x}_i) - h(-\hat{x}_i)\right),
\end{align*}
so that by the Hoeffding inequality, any $\alpha > 0$ satisfies
\begin{align*} 
\Prx_{\bz \sim \{-1,1\}^n}\left[ \left| \Ex_{\by \sim D_{\bz}}[h(\by)] - \Ex_{\by \sim D}[h(\by)]\right| \geq \alpha \right] &\leq \exp\left(- \frac{2\alpha^2 t_2(x)^2 \|x\|_{L_2(X)}^2}{\eps_0^2 \sum_{i=1}^n \|x_i\|_X^2 (h(\hat{x}_i) - h(-\hat{x}_i))^2}  \right). \\
	&\leq \exp\left(- \Omega\left(\frac{\alpha^2 t_2(x)^2}{\eps_0^2}\right) \right),
\end{align*}
where we used the fact that $1 \leq \|x_i\| \leq 2$, as well as the fact that $\|h\|_D = 1$ to say that $\|x\|_{L_2(X)}^2 \gsim \sum_{i=1}^n \|x_i\|_X^2(h(\hat{x}_i) - h(-\hat{x}_i))^2$.
Let $Z \subset \{-1,1\}^n$ be any subset of size $|Z| \geq 2^d / r$, and let $\calD_Z = \{ D_z : z \in Z\} \subset \calD$ be the corresponding set of distributions, and so, similarly to the proof of Lemma~3.21 in \cite{FGV17},
\begin{align*}
\Prx_{\bz \sim Z}\left[\left| \Ex_{\by \sim D}[h(\by)] - \Ex_{\by \sim D_{\bz}}[h(\by)]\right| \geq \alpha \right] \leq r \exp\left( - \Omega\left(\frac{\alpha^2 t_2(x)^2}{\eps_0^2} \right)\right),
\end{align*}
which implies
$\Ex_{\bz \sim Z}\left[\left| \Ex_{\by \sim D}[h(\by)] - \Ex_{\by \sim D_{\bz}}[h(\by)]\right| \right] \lsim \frac{\eps_0 \sqrt{\ln r}}{t_2(x)}
$.
Then, for any $\eps \le \eps_0$, any subset of $\calD$ containing at least $\exp(- O(\kappa^2 t_2(x)^2 / \eps^2 ))$-fraction of distributions will have expectation within $\kappa$ of $\Ex_{\by\sim D}[h(\by)]$.
\end{proof}
Combining Lemma~\ref{lem:statis-dim}, Lemma~\ref{lem:dist}, and Theorem~\ref{thm:fpv}, we obtain a proof of Theorem~\ref{thm:symm-lb}.


\section{Lower bounds for Schatten-$p$ norms}

For the remainder of the section, $S_p = (\R^{d\times d}, \|\cdot\|_{S_p})$ is the Schatten-$p$ normed space, defined over the vector space of $d \times d$ matrices, and $\|x\|_{S_p} = (\sum_{i=1}^d |\sigma_i(x)|^{p})^{1/p}$ where $\sigma_i(x)$ is the $i$-th singular value of $x$. 
By a straightforward calculation, the following upper bound holds by embedding into $\ell_2^{d \times d}$ via the identity map, and applying SQ mean estimation algorithm for $\ell_2$:
\begin{corollary}
\label{cor:schatten-ub}
There exists a statistical algorithm for mean estimation in $S_p$ making $d^{O(1)}$-queries to $\STAT(\alpha)$ with 
\[ \alpha = \Omega \left( \frac{\eps}{d^{1/2 - 1/p}}\right). \]
\end{corollary}

\noindent
The rest of this section is dedicated to showing the following lower bound, which yields the corresponding lower bound to Corollary~\ref{cor:schatten-ub}.
\begin{lemma}
\label{lem:schatten-lb}
There exists an $\eps > 0$ such that any SQ algorithm for mean estimation in $S_p$ with error $\eps$ making queries to $\VSTAT(1/(3\kappa^2))$ must make $\exp( \Omega( \min\{\frac{\kappa^2  d^{1-2/p}}{\eps^2}, d + \log \kappa \}) )$ queries.
\end{lemma}
\noindent
Similarly to Theorem~\ref{thm:symm-lb}, we obtain the following, which shows that Corollary~\ref{cor:schatten-ub} is optimal.
\begin{theorem}
\label{schatten_main}
Any statistical algorithm for mean estimation in $S_p$ making $d^{O(1)}$-queries to $\STAT(\alpha)$ must have 
\[ \alpha = O\left( \frac{\eps}{d^{1/2 - 1/p}}\right). \]
\end{theorem}

\paragraph{Description of the lower bound instance}
We now describe the instance which achieves the lower bound in Lemma~\ref{lem:schatten-lb}. 
Consider the distribution $D$ supported on $d \times d$ matrices generated by the following process: 1) let $\bpi \sim \calS_d$ be a uniformly random permutation on $[d]$, 2) independently sample $\bz \sim \{-1,1\}^d$, and output the matrix $\by = y(\bpi, \bz) = (y(\bpi, \bz)_{ij}) \in \R^{d\times d}$ where 
\[ y(\bpi, \bz)_{ij} =\left\{ \begin{array}{cc} \bz_i / d^{1/p} & j = \bpi(i) \\
									0		& \text{o.w} \end{array} \right. .\] 
Note that $\by \sim D$ always satisfies $|\sigma_1(\by)| = \dots = |\sigma_d(\by)| = 1/d^{1/p}$, so that $\|\by\|_{S_p} = 1$, and that $\Ex_{\by \sim D}[\by] = 0$.

Let $0 < \eps \leq \gamma d^{1/p}$ be a parameter for a sufficiently small constant $\gamma > 0$. For $a, b \in \{-1,1\}^d$, let $D_{a,b}$ be the distribution supported on $d \times d$ matrices generated by the following process: 1) let $\bpi \sim \calS_d$ be a uniformly random permutation on $[d]$, 2) sample $\bz \sim \{-1, 1\}^d$ where each $i \in [d]$ is independently distributed with $\Prx[\bz_i = a_{i}b_{\bpi(i)}] = \frac{1}{2} + \frac{\eps d^{1/p}}{2}$,
and output the matrix $\by = y(\bpi, \bz)$. 
Similarly to the case with $D$, $\by \sim D_{a,b}$ always satisfies $|\sigma_1(\by)| = \dots = |\sigma_d(\by)| = 1/d^{1/p}$, so that $\|\by\|_{S_p} = 1$. Furthermore, in this case, we have
$\mu_{a, b} = \Ex_{\by \sim D_{a,b}}[\by] = \frac{\eps}{d} \cdot a b^{\intercal}$, and $\|\mu_{a,b}\|_{S_p} = \eps$.
Finally, we let $\calD$ be the set of distributions given by $D_{a,b}$ where $a, b \in \{-1,1\}^d$. Since every distribution in $\calD$ has mean with $S_{p}$ norm at least $\eps$, we obtain the following lemma.
\begin{lemma}
Suppose there is a statistical algorithm for mean estimation with error $\eps$ for making $q(\eps)$ queries of accuracy $\alpha(\eps)$, then $\calB(D, \calD)$ has a randomized statistical algorithm making $q(\eps)$ queries of accuracy $\alpha(\eps)$ succeeding with the same probability.
\end{lemma}
\noindent
Similarly to the case in Section~\ref{sec:type-2-lbs}, we obtain lower bounds on algorithms using statistical queries by giving a lower bound on the statistical dimension of $\calB(D, \calD)$. 

\begin{lemma}
\label{lem:schatten-main}
Let $D$ and $\calD$ be the distribution and the set over $B_{S_p}$ defined above. For $\kappa > 0$,
$\SDN(\calB(D, \calD), \kappa) \geq \exp( \Omega(\min\{ \frac{\kappa^2 d^{1 - 2/p}}{\eps^2}, d + \log \kappa\} )).$
\end{lemma}

\begin{proof}
Let $h \colon B_{S_p} \to \R$ be any function with $\|h\|_D = 1$, and denote the Boolean function $H_h \colon \{-1,1\}^d \times \{-1,1\}^d \to \R$ by:
\begin{align}
H_h(a, b) &= \Ex_{\by \sim D_{a,b}}[h(\by)] - \Ex_{\by \sim D}[h(\by)] \nonumber\\
&= \frac{1}{d!} \sum_{\pi \in \calS_d} \frac{1}{2^d} \sum_{z \in\{-1,1\}^d} h(y(\pi, z)) \left(\prod_{i=1}^d (1 + \eps d^{1/p}z_i a_ib_{\pi(i)}) - 1\right) \nonumber\\
	&= \frac{1}{d!} \sum_{\pi \in \calS_d} \sum_{S \subset [d]} (\eps d^{1/p})^{|S|} \cdot \chi_{S}(ab_{\pi}) \cdot \hat{h_\pi}(S), \label{eq:ajja}
\end{align}
where we write $h_{\pi} \colon \{-1,1\}^d \to [0,1]$ to denote $h_{\pi}(z) = h(y(\pi, z))$, for $S \subset [d]$, $\chi_{S} \colon \{-1,1\}^d \to \{-1,1\}$ is given by $\chi_S(z) = \prod_{i\in S} z_i$, and $ab_{\pi} \in \{-1,1\}^d$ denotes the vector where $(ab_{\pi})_i = a_i b_{\pi(i)}$.
Further consolidating terms, we can write
\begin{equation}
H_h(a, b) = \frac{1}{d!} \sum_{t = 1}^d (\eps d^{1/p})^{t} \sum_{\substack{S, T \subseteq [d]\\|S| = |T| = t}} \Gamma_{S, T} \cdot \chi_{S} (a) \chi_{T} (b) \qquad \text{where} \qquad \Gamma_{S, T} = \sum_{\substack{\pi \in \calS_d:\\\pi(S) = T}} \hat{h_\pi} (S). \label{eq:H-def}
\end{equation}
Similarly to the proof of Lemma~\ref{lem:statis-dim}, we will use a concentration bound on $H(\ba, \bb)$ when $\ba, \bb \sim \{-1,1\}^d$ to derive a bound on the statistical dimension. Specifically, Lemma~\ref{lem:schatten-conc} (which we state and prove next), as well as a union bound, implies that for any $2 \leq q \leq d/(2e)$, and any set of pairs $Z \subset \{-1,1\}^d \times \{-1,1\}^d$ of size at least $2^{2d} / r$, and $\calD_Z = \{ D_{a,b} : (a,b) \in Z\}$, 
\[ \Prx_{(\ba, \bb) \sim Z}\left[ |H_h(\ba, \bb)| \geq \frac{4e \sqrt{q}\cdot \eps}{ d^{1/2-1/p}}\right] \leq r 2^{-q}. \]
We may also apply Cauchy-Schwartz inequality to (\ref{eq:ajja}) to say that for every $a, b \in \{-1,1\}^d$,
\begin{align*}
H_{h}(a, b) &\leq \left(\frac{1}{d!} \sum_{\pi \in \calS_d} \sum_{S \subset [d]} (\eps d^{1/p})^{2|S|} \right)^{1/2} \left( \frac{1}{d!} \sum_{\pi \in \calS_d} \sum_{S \subset [d]} \hat{h_\pi}(S)^2 \right)^{1/2} \\
		&\leq (1 + \eps^2 d^{2/p})^{d/2} \cdot \|h\|_D = (1 + \eps^2 d^{2/p})^{d/2}.
\end{align*}
This, in turn, implies
\[ \Ex_{(\ba, \bb) \sim Z}\left[|H_h(\ba, \bb)| \right] \lsim \frac{\sqrt{\log r} \cdot \eps }{d^{1/2 - 1/p}} + \left(1 + \eps^2 d^{2/p} \right)^{d/2} \cdot r 2^{-d/(2e)}  \lsim \frac{\sqrt{\log r} \cdot \eps}{d^{1/2 - 1/p}} + r \cdot 2^{-d/6} \]
when $\eps$ is a small constant times $d^{-1/p}$.
Therefore, we have $\Ex_Z[|H_h(\ba, \bb)|] \leq \kappa$ for all subsets containing at least $2^{2d} / r$ distributions, where $r = \exp( \Omega( \min \{\frac{\kappa^2  d^{1 - 2/p}}{\eps^2}, d + \log \kappa \} ) )$. 
\end{proof}
\noindent
We now prove the concentration inequality for $H_{h}(\ba, \bb)$ used in the proof of Lemma~\ref{lem:schatten-main}.

\begin{lemma}
\label{lem:schatten-conc}
Let $h \colon B_{S_p} \to \R$ satisfy $\|h\|_D = 1$, and let $H_h \colon \{-1,1\}^d \times \{-1,1\}^d \to \R$ be the function in (\ref{eq:H-def}). 
Then, for any $2 \leq q \leq d/(2e)$,
$
\Prx_{\ba, \bb \sim \{-1, 1\}^d} [ |H_h(\ba, \bb)| > \frac{4 e\sqrt{q} \eps}{d^{1/2 - 1/p}} ] \leq 2^{-q}$.
\end{lemma}
\noindent
To prove this lemma, we setup additional technical machinery.
Recall that for any $\rho \in [-1, \infty)$ the \emph{noise operator} $\cT_\rho$ is the linear operator on Boolean functions, defined so that for any Boolean function $f\colon \{-1, 1\}^m \to \R$ with Fourier expansion $f(x) = \sum_{S \subseteq [m]} \hat{f}(S) \chi_S (x)$ where $\chi_S(x) = \prod_{i \in S} x_i$, we have $\cT_\rho f(x) = \sum_{S \subseteq [m]} \rho^{|S|} \hat{f}(S) \chi_S (x)$.\footnote{The operator $\cT_{\rho}$ is typically only defined for $\rho \in [-1, 1]$, but one may naturally extend this definition to $\rho > 1$, see e.g. \cite{O14}.}
We will use the following version of the hypercontractivity theorem, which will allow us to bound moments of random Boolean functions.

\begin{theorem}[$(2, q)$-Hypercontractivity, Chapter 9 in \cite{O14}]
\label{thm:hypercontractivity}
Let $f\colon \{-1, 1\}^m \to \R$, and let $2 \leq q \leq \infty$.
Then for $\rho = 1/\sqrt{q-1}$,
$\Ex_{\bx \sim \{-1, 1\}^m} \left[ \left|\cT_{\rho} f(\bx) \right|^q \right] \leq \Ex_{\bx \sim \{-1, 1\}^m} \left[ f(\bx)^2 \right]^{q/2}
$.
\end{theorem}

\begin{proof}[Proof of Lemma~\ref{lem:schatten-conc}]
Define the auxiliary Boolean function $g \colon \{-1,1\}^d \times \{-1,1\}^d \to \R$ by
\[
g(a,b) = \frac{1}{d!} \sum_{t = 1}^d \sum_{\substack{S, T \subseteq [d]\\|S| = |T| = t}} \Gamma_{S, T} \cdot \chi_{S} (a) \chi_{T} (b),
\]
for $\Gamma_{S,T}$ as in (\ref{eq:H-def}). Note that for $\sigma = \sqrt{\eps d^{1/p} (q-1)}$ and $\rho = 1/\sqrt{q-1}$, $H_h(a, b) = \cT_{\rho} \cT_{\sigma} g(a, b)$. For all $2 \leq q \leq \infty$, we have
\begin{align}
\Prx_{\ba,\bb \sim \{-1, 1\}^d} \left[ \left|H_h(\ba, \bb)\right| > \alpha \right] &\leq \frac{\Ex_{\ba, \bb \sim \{-1, 1\}^d} \left[ |H_h(\ba, \bb)|^q \right]}{\alpha^q} \nonumber \\
&\leq \frac{\Ex_{\ba, \bb \sim \{-1, 1\}^d} \left[ \cT_{\sigma} g(\ba, \bb)^2 \right]^{q/2}}{\alpha^q} \; , \label{eq:hypercontractivity}
\end{align}
where the first inequality follows from Markov's inequality, and the second from $(2,q)$-hypercontractivity (Theorem~\ref{thm:hypercontractivity}). By Parseval's identity, observe that
\begin{align*}
\Ex_{\ba, \bb \sim \{-1, 1\}^d} \left[ \cT_{\sigma} g(a, b)^2 \right] &\leq  \eps^2 d^{2/p} \left( \frac{1}{d!} \right)^2  \sum_{t = 1}^d q^t \sum_{\substack{S, T \subseteq [d]\\|S| = |T| = t}} \Gamma_{S, T}^2.
\end{align*}
For any fixed $1 \leq t \leq d$, recall from (\ref{eq:H-def}) that 
\begin{align*}
\sum_{\substack{S, T \subseteq [d]\\|S| = |T| = t}} \Gamma_{S, T}^2 &= \sum_{\substack{S, T \subseteq [d]\\|S| = |T| = t}} \left( \sum_{\substack{\pi \in \calS_d\\\pi(S) = T}} \hat{h_\pi} (S) \right)^2 \stackrel{(a)}{\leq} (d - t)! \sum_{\substack{S, T \subseteq [d]\\|S| = |T| = t}} ~ \sum_{\substack{\pi \in \calS_d\\\pi(S) = T}} \hat{h_\pi} (S)^2 \\
&= (d - t)! \sum_{\pi \in \calS_d} \sum_{\substack{S \subseteq [d]\\|S| = t}} \hat{h_\pi} (S)^2 \stackrel{(b)}{\leq} (d - t)! d! \; ,
\end{align*}
where (a) follows by Cauchy-Schwarz, and (b) follows since 
$\frac{1}{d!} \sum_{\pi \in \calS_d} \sum_{S \subset [d]} \hat{h_{\pi}}(S)^2 = 1$, as $\|h\|_D = 1$. 
Summing over all $t \in [d]$, we have
\begin{align}
\Ex_{\ba, \bb \sim \{-1, 1\}^d} \left[ \cT_{\sigma} g(\ba, \bb)^2 \right] &\leq \eps^2 d^{2/p} \left( \frac{1}{d!} \right)^2 \sum_{t = 1}^d q^t (d - t)! d! = \eps^2 d^{2/p} \sum_{t = 1}^d \frac{q^t (d - t)!}{d!} \nonumber\\
& = q \eps^2 d^{2/p - 1} \sum_{t = 0}^{d - 1} \frac{q^t (d - t - 1)!}{(d - 1)!}, \label{eq:h1_bound_1}
\end{align}
and using Stirling's approximation, 
\begin{align*}
\sum_{t = 0}^{d - 1} \frac{q^t (d - t - 1)!}{(d - 1)!} &\leq \sum_{t = 0}^{d - 1} e q^t  \sqrt{\frac{ d - t - 1}{d - 1}}  \left( \frac{(d - t - 1)}{e} \right)^{d - t- 1}  \left( \frac{e}{d - 1} \right)^{d - 1} \\
&\leq e \sum_{t = 0}^{d - 1} \left( \frac{eq}{d} \right)^{t} \leq 2e \; ,
\end{align*}
for all $q \leq d/(2e)$.
Therefore~\eqref{eq:h1_bound_1} simplifies to give
$\Ex_{\ba, \bb \sim \{-1, 1\}^d} \left[ \cT_{\sigma}g(\ba, \bb)^2 \right] \leq 2 e q \eps^2 d^{2/p - 1}$,
for all $q \leq d/(2e)$, and plugging this bound into~\eqref{eq:hypercontractivity} while letting $\alpha = 4 e \sqrt{q} \cdot \eps / d^{1/2 - 1/p}$, we obtain the desired concentration bound.
\end{proof}

\begin{flushleft}
\bibliographystyle{alpha}
\bibliography{waingarten}
\end{flushleft}

\end{document}